\documentclass[conference]{IEEEtran}
\usepackage{graphicx}  
\usepackage{amsmath, amsthm, amsfonts}
\usepackage{float}
\usepackage[linesnumbered,ruled,vlined]{algorithm2e}
\usepackage{xcolor}
\usepackage{tikz}
\usepackage{pifont}

\newtheorem{theorem}{Theorem}
\newtheorem{lemma}{Lemma}
\theoremstyle{definition}

\newtheorem{remark}{Remark}

\usepackage{wrapfig}
\usepackage{mathtools, cuted}
\usepackage{amstext}
\usepackage{amssymb}
\usepackage{mathrsfs}
\usepackage{mathtools}
\usepackage{tikz}
\usepackage{marginnote}
\usepackage{tkz-tab}
\usetikzlibrary{topaths,calc}

\allowdisplaybreaks

\newcommand{\remove}[1]{}

\newcommand{\cP}{\mathcal{P}}

\begin{document}
\title{Constrained Maximum Entropy Contiguous Aggregations}

\author{\IEEEauthorblockN{Roberto Bruno \qquad Ugo Vaccaro} \IEEEauthorblockA{Department of Computer Science, University {of Salerno}, Fisciano (SA), Italy\\
Emails: \tt \{rbruno, uvaccaro\}@unisa.it}}

\maketitle

\begin{abstract}
Given a probability distribution $p = (p_1, \dots, p_n)$ and an integer $1\leq m \leq n$, 
a contiguous aggregation of $p$ is a probability distribution
$q = (q_1, \dots, q_m)$ 
 such that each $q_i$ is a sum of consecutive elements of $p$. Given $p$ and a 
 positive number $R$,
we consider the problem of computing
a maximum entropy contiguous aggregation $q$ of $p$, under the constraint that 
its Shannon entropy $H(q)$ is at most $R$.
We devise a dynamic programming algorithm that solves the problem exactly,
and two time-efficient greedy algorithms that provide close-to-optimal
solutions. We discuss a few scenarios where our problem arises.
\end{abstract}

\section{Introduction}
Let 
    $\cP_n = \{p=(p_1,\dots,p_n): p_i\geq 0,  \sum_{i=1}^n p_i = 1\}$
be the $(n-1)$-dimensional probability simplex.
For any $p\in \cP_n$ we denote by $H(p)=-\sum_{i=1}^n p_i\log_2 p_i$ the Shannon entropy of the probability distribution $p$.
Given two probability distributions $p\in \cP_n$ and $q\in \cP_m$, where {$m\leq n$}, we say that $q$ is a \textit{contiguous aggregation} of $p$ if there exist indices $0=i_0<i_1<\dots<i_{m-1}<i_m = n$ such that for each $j=1,\dots,m$ we have $q_j = \sum_{k=i_{j-1}+1}^{i_j} p_k$. 
For any $i, j\in \{1, \ldots, n\}$, $i>j$,
 we call the sequence of consecutive elements $p_{j}, \ldots , p_i$  
a \textit{segment} of $p.$
Let $CA(p)$ denote the set of all probability distributions that are contiguous aggregations of $p$, and $R$ be a positive number.
The optimization problem  we are interested in consists of computing:
\begin{equation}\label{eq:opt_prob}
    \max_{q \in CA(p)} H(q) \quad 
    \text{subject to} \quad H(q)\leq R.
\end{equation}

\subsection{Our Results}
In Section \ref{sec:Dyn}, we present a Dynamic Programming approach to compute the solution
of the optimization problem (\ref{eq:opt_prob}). Unfortunately, the algorithm has 
$\Theta(2^n)$-time-complexity in the worst case. Successively, in Section \ref{sec:Greedy}
we present two time-efficient approximation algorithms for the problem  (\ref{eq:opt_prob}).
The first algorithm returns a value that differs from the optimum by at most 
$h(p^{\downarrow}_{\lfloor2^R\rfloor+1})$,
    where $h(\cdot)$ is the binary entropy, and $p^{\downarrow}_{\lfloor2^R\rfloor+1}$ is the $(\lfloor2^R\rfloor+1)$-th largest probability in $p$. The second algorithm returns a value that differs from the optimum by at most 
$1/({\lfloor({\lfloor2^R\rfloor-1})/{2}\rfloor+1}).$ We also analyze both algorithms in terms
of their multiplicative approximation factor. The latter analysis is more useful than the former 
---done 
in terms of the additive factor--- when the optimal solution to (\ref{eq:opt_prob})
is very small.

\section{Motivations and Related Work}
The problem of identifying \lq\lq informative" contiguous aggregations of a probability distribution appears in several areas of information theory and data processing. The basic motivation stems from the need to reduce the size of a source alphabet or the entropy of the source
while preserving as much statistical information as possible.
In the field of source coding, the technique known as Alphabet Partitioning is used to speed up 
 coding algorithms for large alphabets by grouping symbols into clusters (see  
\cite{bar, bar2, chen, Moffatt}, \cite[Chapt.4]{PS},  \cite{Rya}, \cite{Said2, Said3} and references 
therein quoted).
However, these works ignore the need to preserve the contiguity of symbols, which is crucial in
many applications. Some of these 
applications include alphabetic coding \cite{BDPV, BDDPV, cicalese2018maximum}, design of
scalar quantizers \cite{KK99}, and histogram segmentation \cite{muresan2008quantization}.
Yet, these last works evaluate the goodness of the contiguous aggregations according to
criteria different from ours. In particular, the paper \cite{muresan2008quantization}
starts by considering the problem of identifying a contiguous partition $\cal P$
that minimizes a distortion measure $D({\cal P})$, under the constraint that 
the entropy $H({\cal P})$ induced by $\cal P$ is bounded by some $R$, and 
provides algorithms to minimize the associated Lagrangian $J({\cal P}, \lambda)=D({\cal P})+\lambda H({\cal P})$.
Our criterion (\ref{eq:opt_prob}) can be seen as that of 
maximizing the mutual 
information between the original information source and its compressed version, under the constraint that the compressed version has entropy bounded by an input parameter $R$.
Under this point of view, our problem is in the same line of research as 
\cite{ebrahimi2024minimum, muresan2008quantization,Ng, slo, ss}.  We also point out that 
 the same problem as (\ref{eq:opt_prob}) has been 
studied in \cite{ebrahimi2024minimum} and \cite{noi}, but without the contiguity
constraint considered here.
Finally, we remark that the problem of partitioning data (or probability distributions) into
contiguous segments, optimizing  several information theoretic measures, 
has been studied in multiple other contexts, e,g., see  \cite{CK, ER, PM86} and references therein quoted.
 
\remove{A primary application of contiguous aggregation is found in the design of scalar quantizers. In this context, the source symbols are typically ordered (e.g., signal intensities or histogram bins), and the quantization process consists of grouping consecutive values into a smaller set of discrete levels. Kampke and Kober \cite{KK99} explore these discrete signal quantizations, emphasizing the structural properties required to maintain fidelity in discrete representations.

More recently, Muresan and Effros \cite{muresan2008quantization} framed quantization as a problem of histogram segmentation. They demonstrate that for many network-based systems, optimal scalar quantizer design is equivalent to finding a contiguous partition of a source histogram that optimizes a specific performance metric. Our problem generalizes this by seeking a maximum entropy representation under a specific entropic budget $R$ which can be viewed as finding the "most informative" aggregation that stays within a prescribed complexity or rate constraint.
}
 
\section{A Dynamic Programming Algorithm}\label{sec:Dyn}
Firstly, we would like to point out that a \lq\lq natural'' Dynamic Programming (DP)
algorithm 
does not always provide a solution to the optimization problem (\ref{eq:opt_prob}). 
One could think to proceed as follows: For each value $k\in\{1, \ldots, n\}$, compute and store
the \textit{maximum-entropy} contiguous aggregation of $p=(p_1,\dots,p_n)$
consisting of \textit{exactly} $k$ contiguous segments of  $p$; 
this computation can be performed in overall time $O(n^4)$ using the techniques of 
\cite{cicalese2018maximum}.
Subsequently,
output the $m$-segments aggregation $\overline{q}$, for some $m\in \{1, \ldots, n\}$, that, 
among all previously computed 
and stored aggregations \textit{satisfying} the bound
$H(q)\leq R$, has maximum entropy. However, a moment's thought reveals that there could be 
another $m'\neq m$-segments aggregation $q'$ of $p$ such that  $H(q')>H(\overline{q})$,
that has \textit{not} been stored since \textit{is not} 
a maximum entropy aggregation. That is,  there exists an  
$m'$-segments aggregation $q''$ of $p$ such that $H(q'')>H(q')$, \textit{but} $q''$ 
is not 
admissible 
since 
 $H(q'')>R$.
A correct DP algorithm is the following. Let 
$P_i = \sum_{k=1}^i p_k, \mbox{ for } i=0, \dots, n,$ where we set $P_0 = 0$ by convention.
Notice that for $i, j\in \{1, \ldots, n\}$, $i>j$
we have that  $P_i - P_j=\sum_{k=j+1}^ip_k$.
We define the entropic contribution of the single segment $p_{j+1}, \ldots , p_i$
as 
$E(j, i) = - (P_i - P_j) \log_2 (P_i - P_j).$
For each $k\leq i$
 we define
\begin{align}\label{eq:DP}
  \!\!\!\!\!DP[k,i]=&\{h \mid \exists \mbox{ a partition $q$ in $k$ segments of}\
      p_1, \ldots , p_i\nonumber \\
      &\qquad \mbox{  such that the entropy of $q$ is }
 \!=h \leq R\}.
\end{align}
Here, for convenience, we call \lq\lq entropy" any quantity of the form $-\sum_{a=1}^s
q_a\log_2q_a$, for arbitrary positive numbers $q_1, \ldots, q_s$  such that $\sum_{a=1}^sq_a\leq 1$.
For $i=1, \ldots, n$,  the base cases of (\ref{eq:DP}) are 
\[
DP[1,i] = \begin{cases} \{ E(0,i)\} & \text{if } E(0,i) \le R, \\ \emptyset & \text{otherwise.} \end{cases}
\]
All values $DP[k,i]$ can be computed by the recurrence relation
\begin{align}\label{DP}
DP[k,i] =& \bigcup_{j=k-1}^{i-1} \Big\{ v + E(j, i) \mid v \in 
DP[k-1,j] \nonumber\\
& \qquad\quad \ \ \text{ AND } (v + E(j, i) \le R) \Big\}.
\end{align}
The solution to (\ref{eq:opt_prob})
can be computed by
\[
H_{opt} = \max \left( \bigcup_{k=1}^{n} DP[k,n] \right).
\]
For each allowed value of $k$ and $i$,
one can see that the cardinality of $DP[k,i]$ could be as large as 
$\binom{i-1}{k-1}$. This happens in the unlikely (but 
possible) case when \textit{each} different partition 
generates a different entropy value. If this unfavorable case occurs
often, 
then  the cardinality 
of $\bigcup_{k=1}^{n} DP[k,n]$ can  be as large as $\Theta(2^n)$; therefore, the 
 DP algorithm based on (\ref{DP}) has $\Theta(2^n)$ worst-case time and space complexity.

\section{Efficient Greedy Approximation Algorithms}\label{sec:Greedy}
In this section, we present two greedy algorithms for computing near-optimal 
contiguous aggregations. To avoid trivialities, we assume $R<\log_2 n$. 
Indeed, since the entropy $H(p)$ satisfies $H(p)\leq \log_2 n$, for $R\geq \log_2 n$ the optimal solution would be $p$ itself. We recall that  
 $p^{\downarrow}_i$ denotes the $i$-th largest probability of the distribution $p$.
\remove{In the algorithms we propose, we follow an approach 
inspired from \cite{ebrahimi2024minimum, noi}. 
Specifically, 
we design an algorithm that finds a contiguous aggregation of $p$ whose entropy is at most $h(p^{\downarrow}_{\lfloor2^R\rfloor+1})$ bits away from the \textit{optimal} one, where $h(x)=-x\log_2 x -(1-x)\log_2 (1-x)$ is the binary entropy function and $p^{\downarrow}_{\lfloor2^R\rfloor+1}$ is the $(\lfloor2^R\rfloor+1)$-th largest probability in $p$.
The pseudo-code of the algorithm is presented in Algorithm \ref{alg:version_1}.}
The algorithm constructs a sequence of contiguous aggregations $p^{(0)},\dots,p^{(n-1)}$ 
of decreasing entropy,  and outputs the first aggregation whose entropy does not exceed $R$. It starts by computing the indexes $i_1,\dots,i_{\lfloor2^R\rfloor}$ of the $\lfloor2^R\rfloor$ largest probabilities in $p$. Then, at each step, the algorithm aggregates a pair of contiguous probabilities with the constraint that it \textit{never} merges a pair in which both masses contain the initial probabilities in $\{i_1,\dots,i_{\lfloor2^R\rfloor}\}$. Specifically, starting from the rightmost large probability $p_{i_{\lfloor2^R\rfloor}}$, the algorithm aggregates it, step by step, with the adjacent probabilities on its right (if any). Once no more rightward merges are possible, it proceeds to merge {the mass containing $p_{i_{\lfloor2^R\rfloor}}$}  with the adjacent probabilities on its left while avoiding merging any of the $\lfloor2^R\rfloor$ largest masses with each other. In fact, as soon as the algorithm encounters $p_{i_{\lfloor2^R\rfloor-1}}$, it proceeds to aggregate $p_{i_{\lfloor2^R\rfloor-1}}$ with the adjacent probabilities to its left until it finds the next large probability, $p_{i_{\lfloor2^R\rfloor-2}}$. Thus, this pattern repeats itself until the algorithm finds the first contiguous aggregation whose entropy does not exceed $R$.
The pseudo-code of the algorithm is presented in Algorithm \ref{alg:version_1}.

 \vspace*{-.1cm}
\begin{algorithm}
\small
\caption{AvoidMaxAggregation}\label{alg:version_1}
\KwIn{A probability distribution $p=(p_1,\dots,p_n)\in \cP_n$ and a real value $R$}
\If{$H(p)\leq R$}{
    \Return $p$
}
Let $i_1<i_2<\dots<i_{\lfloor 2^R\rfloor}$ be the indexes of the $\lfloor2^R\rfloor$ largest probabilities of $p$

$p^{(0)} \gets p$

\For{$j \gets 1$ \textbf{to} $n-i_{\lfloor2^R\rfloor}$}{
    $p^{(j)} \gets (p^{(j-1)}_{1},\dots,p^{(j-1)}_{i_{\lfloor2^R\rfloor}}+p^{(j-1)}_{i_{\lfloor2^R\rfloor}+1}, p^{(j-1)}_{i_{\lfloor2^R\rfloor}+2},\dots)$

    \If{$H(p^{(j)})\leq R$}{
        \Return $p^{(j)}$
    }
}

$p^{(n-i_{\lfloor2^R\rfloor})} \gets (p_1,\dots,p_{i_{\lfloor2^R\rfloor}-1}, \sum_{k=i_{\lfloor2^R\rfloor}}^n p_k)$

$\ell \gets \lfloor2^R\rfloor$

\For{$j \gets n-i_{\lfloor2^R\rfloor}+1$ \textbf{to} $n-1$}{
    \If{$n-j == i_{\ell-1}$}{
        $p^{(j)}\gets p^{(j-1)}$
        
        $\ell \gets \ell-1$
        
        $j \gets j+1$
    }
    $p^{(j)} \gets (p^{(j-1)}_1,\dots, p^{(j-1)}_{n-j-1}, p^{(j-1)}_{n-j} + p^{(j-1)}_{n-j+1}, \dots)$

    \If{$H(p^{(j)})\leq R$}{
        \Return $p^{(j)}$
    }
}

\end{algorithm}

\vspace*{-.2cm}
Algorithm \ref{alg:version_1} has $O(n)$ 
time complexity, where $n$ is the size of the probability distribution $p$.
This is because the computation of the $(n-\lfloor2^R\rfloor)$-th smallest element of a set
of $n$ elements can be performed in $O(n)$-time (see Sec. 9.3 of \cite{CLRS}). 
Successively, one can scan $p$ to identify the indices of the $\lfloor2^R\rfloor$ largest probabilities.
\remove{$O(nR)$ time complexity, where $n$ is the size of the probability distribution $p$ and $R<\log_2 n$. This is because the computation of the indices of the $\lfloor2^R\rfloor$ largest probabilities can be performed in $O(nR)$ time using a heap data structure of size $\lfloor2^R\rfloor$ to keep track of them.}
Furthermore, the two loops in Lines 5 and 11 require at most $n-1$ steps, and each step requires constant time since the entropy calculation can be done in constant time by computing only the decrease in entropy after merging the two probabilities in the distribution.
Before demonstrating that the entropy of the aggregation produced by Algorithm \ref{alg:version_1} is at most $h(p^{\downarrow}_{\lfloor 2^R \rfloor + 1})$ bits away from the optimal, we recall the following 
result from \cite{ebrahimi2024minimum,Sz}, which describes how the entropy of a distribution decreases when two probabilities are aggregated.

\begin{lemma}[\cite{ebrahimi2024minimum,Sz}]
\label{lemma:DeltaH}
Let $p=(p_1,\dots,p_n)\in \cP_n$ be a probability distribution and let $q$ be an aggregation of $p$ obtained by merging two probabilities, $p_i$ and $p_j$ with $i<j$. Then, the amount of decrease 
between the entropy of  $p$ and $q$ 
\begin{align*}
    \Delta_H(p_i,p_j) =& H(p)-H(q)\\
                    =& p_i \log_2 \frac{1}{p_i} + p_j\log_2 \frac{1}{p_j} -(p_i+p_j) \log_2\frac{1}{p_i+p_j},
\end{align*}
satisfies the following properties:
\begin{enumerate}
    \item $\Delta_H(\cdot,\cdot)$ is monotonically increasing in both arguments;
    \item For any $x\in (0,1)$ it holds that 
    $$\Delta_H(x,1-x) = -x\log_2 x -(1-x)\log_2 (1-x) = h(x).$$
\end{enumerate}
\end{lemma}
We now prove the approximation guarantee of {Algorithm 1}. During the analysis, ties 
in the relation order among
the probabilities are broken arbitrarily.

\begin{theorem}\label{th:avoid_max_add}
    Given a probability distribution $p=(p_1,\dots,p_n)\in \cP_n$, let us denote by $OPT(p,R)$ the optimal solution to the optimization problem defined in (\ref{eq:opt_prob}), and let $h(\cdot)$ be the binary entropy function. Algorithm \ref{alg:version_1} produces a contiguous aggregation $q$ of $p$ that satisfies
    \begin{equation}
        OPT(p,R)-H(q) \leq h(p^{\downarrow}_{\lfloor2^R\rfloor+1}),
    \end{equation}
    where $p^{\downarrow}_{\lfloor2^R\rfloor+1}$ is the $(\lfloor2^R\rfloor+1)$-th largest probability in $p$.
\end{theorem} 
\begin{proof}
    Algorithm \ref{alg:version_1} produces a sequence of contiguous aggregations $p^{(0)},p^{(1)},\dots, p^{(n-1)}$, and outputs the first aggregation $p^{(i)}$ for which it holds that $H(p^{(i)})\leq R$ and $H(p^{(i-1)})>R$. Let $p^{(i-1)}_k$ and $p^{(i-1)}_{k+1}$ be the pair of probabilities that are aggregated to produce $p^{(i)}$. Thus, it holds that
    \begin{align}
        OPT&(p,R) - H(p^{(i)}) \leq R - H(p^{(i)})\nonumber\\
        &\qquad\qquad\qquad \qquad\qquad \mbox{(since $OPT(p,R)\leq R$)}\nonumber\\
        &< H(p^{(i-1)})-H(p^{(i)}) \quad\mbox{(since $H(p^{(i-1)})>R$)}\nonumber\\
        &= \Delta_H(p^{(i-1)}_k,p^{(i-1)}_{k+1})\label{eq:delta_function}\\
        &\!\!\!\!\!\!\!\!\!\!\mbox{(since $p^{(i-1)}_k$ and $p^{(i-1)}_{k+1}$} 
        \mbox{ is the pair that is merged together)} \nonumber
    \end{align}
    An important observation is that Algorithm \ref{alg:version_1} \textit{never} merges 
    with each other the $\lfloor2^R\rfloor$-th largest probabilities, and this 
    condition holds throughout its execution since the algorithm
     requires \textit{at most} $n-\lfloor2^R\rfloor$ merges to stop.
    In fact, after $n-\lfloor2^R\rfloor$ merges, the size of the obtained contiguous aggregation is $\lfloor2^R\rfloor$, and its entropy is upper bounded by $\log_2 \lfloor2^R\rfloor\leq \log_2 2^R=R$.
    Additionally, at each step, the algorithm merges a pair of probabilities such that:
    \begin{itemize}
        \item the first probability is either one of the $\lfloor2^R\rfloor$-th largest probabilities, or a sum of probabilities that includes exactly one of the $\lfloor2^R\rfloor$-th largest,
        \item the second one is a probability $p_{\ell}$, where $\ell \in \{1,\dots,n\}\setminus\{i_1,\dots,i_{\lfloor2^R\rfloor}\}$, that has not been merged before.
    \end{itemize}
    Thus, let us assume that $p^{(i-1)}_k$ is the probability of the second type, i.e., $p^{(i-1)}_k = p_{\ell}$ for some $\ell \in\{1,\dots,n\}\setminus\{i_1,\dots,i_{\lfloor2^R\rfloor}\}$. We can rewrite (\ref{eq:delta_function}) as follows
    \begin{align}
        OPT&(p, R) - H(p^{(i)}) < \Delta_H(p^{(i-1)}_k,p^{(i-1)}_{k+1})\nonumber\\
        &=\Delta_H(p_{\ell},p^{(i-1)}_{k+1})\nonumber\\
        &\leq \Delta_H(p_{\ell},1-p_{\ell})\quad\mbox{(from 1. of Lemma \ref{lemma:DeltaH})}\nonumber\\
        &= h(p_{\ell})\quad\mbox{(from 2. of Lemma \ref{lemma:DeltaH})}\nonumber\\
        &\leq \max_{\ell \in \{1,\dots,n\}\setminus\{i_1,\dots,i_{\lfloor2^R\rfloor}\}} h(p_{\ell})
        \leq h(p^{\downarrow}_{\lfloor2^R\rfloor+1}).\nonumber\\
        &\mbox{(since $0.5\geq p^{\downarrow}_{\lfloor2^R\rfloor+1}\geq \max_{\ell\in \{1,\dots,n\}\setminus\{i_1,\dots,i_{\lfloor2^R\rfloor}\}} p_{\ell}$}\nonumber\\
        &\quad\mbox{and $h(x)$ is increasing in $(0,1/2]$).}\nonumber
    \end{align}
\end{proof}

\vspace*{-.3cm}
Theorem \ref{th:avoid_max_add} establishes an additive approximation to the quality of the solution returned by Algorithm \ref{alg:version_1}. We can also provide a multiplicative approximation for it, as shown in the following theorem.

\vspace*{-.2cm}
\begin{theorem}\label{th:avoid_max_mul}
     Given a probability distribution $p=(p_1,\dots,p_n)\in \cP_n$.
     Algorithm \ref{alg:version_1} produces a contiguous aggregation $q$ of $p$ that satisfies
    \begin{equation}
       H(q)>\left(1-\frac{h(p^{\downarrow}_{\lfloor2^R\rfloor+1})}{R}\right)OPT(p,R),
    \end{equation}
    where $p^{\downarrow}_{\lfloor2^R\rfloor+1}$ is the $(\lfloor2^R\rfloor+1)$-th largest probability in $p$.
\end{theorem}
\begin{proof}
   Let us recall that  Algorithm \ref{alg:version_1} returns the first contiguous aggregation $p^{(i)}$ in the sequence $p^{(0)}, p^{(1)},\dots, p^{(n-1)}$ for which it holds that $H(p^{(i)})\leq R$ and $H(p^{(i-1)})>R$. Let $p^{(i-1)}_k$ and $p^{(i-1)}_{k+1}$ be the pair of probabilities that are merged to produce $p^{(i)}$. Thus, we have that
    \begin{align}
        \frac{H(p^{(i)})}{OPT(p,R)}&\geq \frac{H(p^{(i)})}{R}\quad\mbox{(since $R\geq OPT(p,R)$)}\nonumber\\
        >&\frac{H(p^{(i)})}{H(p^{(i-1)})}\quad\mbox{(since $H(p^{(i-1)})>R$)}\nonumber\\
        =&\frac{H(p^{(i-1)})-(H(p^{(i-1)})-H(p^{(i)}))}{H(p^{(i-1)})}\nonumber\\
        =&\frac{H(p^{(i-1)})- \Delta_H\left(p^{(i-1)}_k, p^{(i-1)}_{k+1}\right)}{H(p^{(i-1)})}\nonumber\\
        =& 1 - \frac{\Delta_H\left(p^{(i-1)}_k, p^{(i-1)}_{k+1}\right)}{H(p^{(i-1)})}.\label{eq:mul_bound_step_1}
    \end{align}
    As argued in the proof of Theorem \ref{th:avoid_max_add}, Algorithm \ref{alg:version_1} never merges with each other the $\lfloor2^R\rfloor$-th largest probabilities. Thus, at each step, the algorithm merges a pair of probabilities such that one is a probability containing one of the initial $\lfloor2^R\rfloor$-th largest probabilities, and the other is a probability $p_{\ell}$ that has not been previously merged, where $\ell \in \{1,\dots,n\}\setminus\{i_1,\dots,i_{\lfloor2^R\rfloor}\}$.
    Assuming that 
    $p^{(i-1)}_k = p_{\ell}$ for some $\ell \in\{1,\dots,n\}\setminus\{i_1,\dots,i_{\lfloor2^R\rfloor}\}$, we can rewrite (\ref{eq:mul_bound_step_1}) as follows

    \vspace*{-.4cm}
    \begin{align}
        \frac{H(p^{(i)})}{OPT(p,R)}
         &>1 - \frac{\Delta_H\left(p^{(i-1)}_k, p^{(i-1)}_{k+1}\right)}{H(p^{(i-1)})}\nonumber\\
       & =1 - \frac{\Delta_H\left(p_{\ell},p^{(i-1)}_{k+1}\right)}{H(p^{(i-1)})}\nonumber\\
        &\geq  1 - \frac{\Delta_H\left(p_{\ell},1-p_{\ell}\right)}{H(p^{(i-1)})} \quad\mbox{(from 1. of Lemma \ref{lemma:DeltaH})}\nonumber\\
         &= 1- \frac{h(p_{\ell})}{H(p^{(i-1)})}\quad\mbox{(from 2. of Lemma \ref{lemma:DeltaH}).}\label{eq:mul_bound_step_2}
    \end{align}
    Since $0.5\geq p^{\downarrow}_{\lfloor2^R\rfloor+1}\geq \max_{\ell\in \{1,\dots,n\}\setminus\{i_1,\dots,i_{\lfloor2^R\rfloor}\}} p_{\ell}$ and the function $h(\cdot)$ is increasing in $(0,1/2]$, from (\ref{eq:mul_bound_step_2}) it follows that
    \begin{align}
         \frac{H(p^{(i)})}{OPT(p,R)}>&1- \frac{h(p_{\ell})}{H(p^{(i-1)})}
         \geq 1- \frac{h(p^{\downarrow}_{\lfloor2^R\rfloor+1})}{H(p^{(i-1)})}\nonumber\\
         >&1- \frac{h(p^{\downarrow}_{\lfloor2^R\rfloor+1})}{R}\quad\mbox{(since $R< H(p^{(i-1)})$).}\label{eq:mul_bound_step_3}
    \end{align}
    Finally, rearranging (\ref{eq:mul_bound_step_3}) yields (\ref{eq:conc1}) that concludes the proof.
    \begin{equation}\label{eq:conc1}
        H(p^{(i)})>\left(1-\frac{h(p^{\downarrow}_{\lfloor2^R\rfloor+1})}{R}\right)OPT(p,R).
    \end{equation}
\end{proof}

\vspace*{-.25cm}
We now present an alternative approach leading to a different approximation algorithm, whose pseudocode is illustrated in Algorithm \ref{alg:min_min}. Like the previous one, Algorithm \ref{alg:min_min} iteratively constructs a sequence of contiguous aggregations $p^{(1)},\dots,p^{(n-1)}$ of decreasing entropy. At each step, among all the pairs of consecutive probabilities, Algorithm \ref{alg:min_min} aggregates the one whose sum is the smallest.

\vspace*{-.15cm}
\begin{algorithm}
\small
\caption{Min-Min}\label{alg:min_min}
\KwIn{A probability distribution $p=(p_1,\dots,p_n)\in \cP_n$ and a real value $R$}
\If{$H(p)\leq R$}{
    \Return $p$
}

$p^{(0)} \gets p$

\For{$i \gets 1$ \textbf{to} $n-1$}{
    $p^{(i)} \gets$ Merge the pair of consecutive probabilities in $p^{(i-1)}$ that has the smallest sum 
    
    \If{$H(p^{(i)})\leq R$}{
        \Return $p^{(i)}$
    }
}

\end{algorithm}

\vspace*{-.3cm}
Algorithm \ref{alg:min_min} has $O(n\log n)$ time complexity
since  it executes a loop with at most $O(n)$ steps, and each step can be performed in $O(\log n)$. Indeed, the search for the pair of consecutive probabilities that have the smallest sum can be performed in time $O(\log n)$ using a heap data structure, while the entropy computation can be done in constant time by calculating only the entropy decrease.
The following result holds.

\vspace*{-.2cm}
\begin{theorem}\label{th:min-min_add}
    Given a probability distribution $p=(p_1,\dots,p_n)\in \cP_n$, let us denote by $OPT(p,R)$ the optimal solution to the optimization problem defined in (\ref{eq:opt_prob}). Algorithm \ref{alg:min_min} produces a contiguous aggregation $q$ of $p$ that satisfies
    \begin{equation}
         OPT(p,R)-H(q) < \frac{1}{\lfloor\frac{\lfloor2^R\rfloor-1}{2}\rfloor+1}.
    \end{equation}
\end{theorem}
\begin{proof}
    Algorithm \ref{alg:min_min} produces a sequence of contiguous aggregations $p^{(0)},\dots, p^{(n-1)}$ and outputs the first aggregation $p^{(i)}$ in the sequence for which $H(p^{(i)})\leq R$ and $H(p^{(i-1)})>R$. Let $p^{(i-1)}_{j}$ and  $p^{(i-1)}_{j+1}$ be the two consecutive probabilities that have been merged to obtain $p^{(i)}$. Then, it holds that
    \begin{align}\label{eq:first_step}
     \!\!\!   OPT&(p,R)-H(p^{(i)}) \leq R -H(p^{(i)})\nonumber\\
        &\qquad\qquad \qquad \qquad \quad \ \mbox{(since $OPT(p,R)\leq R$)}\nonumber\\
        <& H(p^{(i-1)})-H(p^{(i)})
        \quad\mbox{(since $R < H(p^{(i-1)})$)}\nonumber\\
        =& \Delta_H(p^{(i-1)}_{j}, p^{(i-1)}_{j+1})\quad\mbox{(by Line 5 of Algorithm \ref{alg:min_min})}.
    \end{align}
    Let $a=p^{(i-1)}_{j}$, $b=p^{(i-1)}_{j+1}$ and $c=a+b$. From (\ref{eq:first_step}), we have
    \begin{align}\label{eq:second_step}
        OPT&(p,R)-H(p^{(i)})<\Delta_H(p^{(i-1)}_{j}, p^{(i-1)}_{j+1})\nonumber\\
        =&a \log_2 \frac{1}{a} + b\log_2 \frac{1}{b} -c\log_2 \frac{1}{c}\nonumber\\
        =& c\left(\frac{a}{c}\log_2 \frac{c}{a}+\frac{b}{c}\log_2 \frac{c}{b}\right)
        = c H\left(\frac{a}{c}, \frac{b}{c}\right)
        \leq c H\left(\frac{1}{2}, \frac{1}{2}\right)\nonumber\\
        =&c= p^{(i-1)}_{j} + p^{(i-1)}_{j+1}.
    \end{align}
    We observe that Algorithm \ref{alg:min_min}, like Algorithm \ref{alg:version_1}, performs at most $n-\lfloor2^R\rfloor$ steps. In fact, the contiguous aggregation $p^{(n-\lfloor2^R\rfloor)}$ has size $\lfloor 2^R\rfloor$ and its entropy is at most $\log_2 \lfloor 2^R\rfloor\leq \log_2 2^R=R$. In addition, since the minimum sum of a pair of consecutive probabilities can only increase in $i$ throughout this sequence of aggregations, from (\ref{eq:second_step}) it follows that
    \begin{equation}\label{eq:third_step}
        OPT(p,R)-H(p^{(i)}) < p^{(n-\lfloor2^R\rfloor-1)}_{k} + p^{(n-\lfloor2^R\rfloor-1)}_{k+1},
    \end{equation}
    where $p^{(n-\lfloor2^R\rfloor-1)}_{k}$ and $p^{(n-\lfloor2^R\rfloor-1)}_{k+1}$ are the pair of consecutive probabilities with the smallest sum in $p^{(n-\lfloor2^R\rfloor-1)}$. 
   We now upper bound (\ref{eq:third_step}). Let $m=n-\lfloor2^R\rfloor-1$, and  
   $p'$ denote the list obtained from $p^{(m)}$ by removing the minimal pair $p^{(m)}_{k}$ and $p^{(m)}_{k+1}$. Further,  let $p''$ be the list obtained by merging all pairs of consecutive elements of $p'$, i.e., $p''=(p'_1+p'_2,p'_3+p'_4,\dots)$. If the size of $p'$ is odd, we ignore the last element since it cannot be merged with any other probability.
    We now prove that every element in $p''$ is greater than or equal to the sum $p^{(m)}_{k} + p^{(m)}_{k+1}$. We first recall that all pairs of consecutive elements in $p'$
    --- except (possibly) the pair $p'_{k-1}=p^{(m)}_{k-1}$ and $p'_{k+1}=p^{(m)}_{k+2}$ ---
    appear as pairs of consecutive elements in $p^{(m)}$ as well. Thus, since $p^{(m)}_{k} + p^{(m)}_{k+1}$ is the smallest sum of a pair of consecutive probabilities in $p^{(m)}$ and $p''$ is obtained from $p'$ by merging the pairs of consecutive elements, it follows that the only possible element of $p''$ that could potentially be smaller than $p^{(m)}_{k} + p^{(m)}_{k+1}$ is the one obtained by merging the pair $p'_{k-1}=p^{(m)}_{k-1}$ and $p'_{k+1}=p^{(m)}_{k+2}$. However, by the definition of $p^{(m)}_{k}$ and $p^{(m)}_{k+1}$ as the pair with the minimal sum, we have that
    \begin{equation}\label{prima}
        p^{(m)}_{k-1} + p^{(m)}_{k}\geq p^{(m)}_{k} + p^{(m)}_{k+1},
    \end{equation}
    and
    \begin{equation}\label{seconda}
        p^{(m)}_{k+1} + p^{(m)}_{k+2} \geq p^{(m)}_{k} + p^{(m)}_{k+1}.
    \end{equation}
    From (\ref{prima}) and (\ref{seconda}) one gets
    \begin{equation}\label{eq:_k-1}
        p^{(m)}_{k-1} \geq  p^{(m)}_{k+1} \quad \mbox{ and } \quad p^{(m)}_{k+2} \geq p^{(m)}_{k}.
    \end{equation}
    Thus, from (\ref{eq:_k-1}) 
    it follows that
    \begin{equation}
        p^{(m)}_{k-1} +p^{(m)}_{k+2} \geq p^{(m)}_{k} + p^{(m)}_{k+1}.
    \end{equation}
    Hence, every element of $p''$ is at least as large  as $p^{(m)}_{k} + p^{(m)}_{k+1}$. 
    
    Let $p'''$ be the list obtained by appending the element $p^{(m)}_{k} + p^{(m)}_{k+1}$ to $p''$. The sum of the elements in $p'''$ is at most 1, and its size is at least $\lfloor(n-m-2)/2\rfloor + 1 =\lfloor(\lfloor2^R\rfloor-1)/2\rfloor+1$. Furthermore, since the sum $p^{(m)}_{k} + p^{(m)}_{k+1}$ is the smallest element in $p'''$ and $m=n-\lfloor2^R\rfloor-1$, one gets that
    \begin{align}\label{eq:up_sum}
        p^{(m)}_{k} + p^{(m)}_{k+1} &= p^{(n-\lfloor2^R\rfloor-1)}_{k} + p^{(n-\lfloor2^R\rfloor-1)}_{k+1}\nonumber\\&\leq \frac{1}{\lfloor\frac{\lfloor2^R\rfloor-1}{2}\rfloor+1}.
    \end{align}
    Finally, combining (\ref{eq:third_step}) and (\ref{eq:up_sum}) yields
    \begin{align*}
         OPT(p,R)-H(p^{(i)}) &< p^{(n-\lfloor2^R\rfloor-1)}_{k} + p^{(n-\lfloor2^R\rfloor-1)}_{k+1}\\
         &\leq \frac{1}{\lfloor\frac{\lfloor2^R\rfloor-1}{2}\rfloor+1}.
    \end{align*}
\end{proof}

\vspace*{-.3cm}
Similarly to Algorithm \ref{alg:version_1}, we can establish a multiplicative guarantee for Algorithm \ref{alg:min_min}, as shown in the following theorem.

\begin{theorem}\label{th:min-min_mul}
    Given a probability distribution $p=(p_1,\dots,p_n)\in \cP_n$.
    Algorithm \ref{alg:min_min} outputs a contiguous aggregation $q$ of $p$ that satisfies
    \begin{equation}
         H(q) > \left(1-\frac{\frac{1}{\lfloor\frac{\lfloor2^R\rfloor-1}{2}\rfloor+1}}{R}\right)OPT(p,R) .
    \end{equation}
\end{theorem}
\begin{proof}
   Let us recall that  Algorithm \ref{alg:min_min} produces a sequence of contiguous aggregations $p^{(0)},\dots, p^{(n-1)}$ and it outputs the first aggregation $p^{(i)}$ in the sequence for which it holds that $H(p^{(i)})\leq R$ and $H(p^{(i-1)})>R$. Let $p^{(i-1)}_{j}$ and  $p^{(i-1)}_{j+1}$ denote the two consecutive probabilities that have been merged to obtain $p^{(i)}$. For this aggregation, it follows that
    \begin{align}\label{eq:min-min_mul_step_1}
        \frac{H(p^{(i)})}{OPT(p,R)}&>\frac{H(p^{(i)})}{H(p^{(i-1)})}\quad\mbox{\big(since $H(p^{(i-1)})>R$\big)}\nonumber\\
        &=\frac{H(p^{(i-1)}) -(H(p^{(i-1)}) - H(p^{(i)}))}{H(p^{(i-1)})}\nonumber\\
        & = 1- \frac{H(p^{(i-1)}) - H(p^{(i)})}{H(p^{(i-1)})}\nonumber\\
         &= 1 - \frac{\Delta_H(p^{(i-1)}_{j}, p^{(i-1)}_{j+1})}{H(p^{(i-1)})}\nonumber\\
        &> 1 -\frac{\Delta_H(p^{(i-1)}_{j}, p^{(i-1)}_{j+1})}{R}\\
       &\qquad\qquad \quad\mbox{\big(since $R<H(p^{(i-1)})$\big)}\nonumber
    \end{align}
    We recall from (\ref{eq:second_step}) that it holds
    \begin{equation}\label{eq:min-min_mul_step_2}
        \Delta_H\left(p^{(i-1)}_{j}, p^{(i-1)}_{j+1}\right)\leq p^{(i-1)}_{j} + p^{(i-1)}_{j+1}.
    \end{equation}
    Furthermore, since Algorithm \ref{alg:min_min} performs at most $n-\lfloor2^R\rfloor$ steps, and using (\ref{eq:third_step}) and (\ref{eq:up_sum}), we get that
    \begin{align}\label{eq:min-min_mul_step_3} 
        p^{(i-1)}_{j} + p^{(i-1)}_{j+1}&\leq p^{(n-\lfloor2^R\rfloor-1)}_{k} + p^{(n-\lfloor2^R\rfloor-1)}_{k+1}\nonumber\\
        &\leq \frac{1}{\lfloor\frac{\lfloor2^R\rfloor-1}{2}\rfloor+1},
    \end{align}
    where $ p^{(n-\lfloor2^R\rfloor-1)}_{k}$ and  $p^{(n-\lfloor2^R\rfloor-1)}_{k+1}$ are the pair of consecutive probabilities with the smallest sum in $ p^{(n-\lfloor2^R\rfloor-1)}$.
   Combining (\ref{eq:min-min_mul_step_1}), (\ref{eq:min-min_mul_step_2}), and (\ref{eq:min-min_mul_step_3}), we obtain
    \begin{align}\label{eq:min-min_mul_final_step}
        \frac{H(p^{(i)})}{OPT(p,R)} &> 1 -\frac{\Delta_H\left(p^{(i-1)}_{j}, p^{(i-1)}_{j+1}\right)}{R}\nonumber\\
        & \geq  1 -\frac{p^{(i-1)}_{j} + p^{(i-1)}_{j+1}}{R}\nonumber\\
        & \geq1 -\frac{ p^{(n-\lfloor2^R\rfloor-1)}_{k} + p^{(n-\lfloor2^R\rfloor-1)}_{k+1}}{R}\nonumber\\
        & \geq1 -\frac{\frac{1}{\lfloor\frac{\lfloor2^R\rfloor-1}{2}\rfloor+1}}{R}.
    \end{align}
    Finally, rearranging (\ref{eq:min-min_mul_final_step}) one gets (\ref{eq:conc2}) 
    that concludes the proof.
    \begin{equation}\label{eq:conc2}
        H(p^{(i)})>\left(1 -\frac{\frac{1}{\lfloor\frac{\lfloor2^R\rfloor-1}{2}\rfloor+1}}{R}\right)OPT(p,R).
    \end{equation}
\end{proof}

\vspace*{-.3cm}
\begin{remark}
     Let us demonstrate that
     Algorithm \ref{alg:version_1} and Algorithm \ref{alg:min_min} are not directly comparable, in the sense that there are cases where Algorithm \ref{alg:version_1} 
     performs better than Algorithm \ref{alg:min_min}, and vice-versa. 
     To illustrate this, consider the  probability distribution 
         $p = (0.25, 0.22, 0.12, 0.20, 0.21).$
     For $R=0.9$, Algorithm \ref{alg:version_1} provides the  contiguous aggregation 
       $ q=(0.79, 0.21),$
     whose entropy $H(q)$ is approximately $0.7415$. 
     On the other hand, Algorithm \ref{alg:min_min} yields the trivial aggregation $q=(1)$,
     with entropy equal to $0$.  Thus, in this case, Algorithm \ref{alg:version_1} outperforms Algorithm \ref{alg:min_min}.
     Consider, now, the uniform probability distribution for $n=10$, that is, 
        $p = (0.1, 0.1, \dots, 0.1, 0.1).$
     For $R=3$, Algorithm \ref{alg:version_1} provides the following aggregation
         $q = (0.1, 0.1, 0.1, 0.1, 0.1, 0.1, 0.1, 0.3)$,
     whose entropy $H(q)$ is approximately $2.846$. In contrast, Algorithm \ref{alg:min_min} outputs 
         $q = (0.2, 0.2, 0.1, 0.1, 0.1, 0.1, 0.1, 0.1),$
     resulting in a higher entropy of $H(q)\approx 2.921$. Thus, in this case, Algorithm \ref{alg:min_min} performs better than Algorithm \ref{alg:version_1}.
\end{remark}

\section{Conclusions and Open Problems}
In this paper, we have derived an exact but inefficient algorithm for the 
constrained entropy maximization problem (\ref{eq:opt_prob}), illustrated
a few scenarios where the problem is  relevant, and designed
efficient approximation algorithms with a good approximation guarantee.
We conjecture
that the maximization problem (\ref{eq:opt_prob}) is NP-hard.
Another interesting extension of our results would be to devise a polynomial-time approximation scheme 
by suitably rounding the input probabilities in the dynamic programming algorithm of Section \ref{sec:Dyn}.

\end{document}